\def\cubicMetric{CM}
\def\operatorRMS[#1]{{\rm rms}\{#1\}}
\def\vNorm{v_{\rm norm}(\timeVar)}
\def\vNormCubic{v^3_{\rm norm}(\timeVar)}

\def\seqGa{\textit{\textbf{a}}}
\def\seqGb{\textit{\textbf{b}}}
\def\seqGc{\textit{\textbf{c}}}
\def\seqGd{\textit{\textbf{d}}}
\def\seqGf{\textit{\textbf{f}}}
\def\seqGg{\textit{\textbf{g}}}
\def\seqGx{\textit{\textbf{x}}}

\def\seqGs{\textit{\textbf{s}}}

\def\seqGyShift[#1]{{\textit{\textbf{y}}}_{#1}}

\def\numberOfClusters{N_{\rm rb}}
\def\RBsize{N_{\rm sc}}
\def\numberOfNulls{N_{\rm null}}

\def\vecyShift[#1]{{\bf y}_{#1}}
\def\shift[#1]{s_{#1}}

\def\eleGa[#1]{{a}_{#1}}
\def\eleGb[#1]{{b}_{#1}}

\def\indexEleOfSeq{i}

\def\upsampleVarA{k}
\def\upsampleVarB{l}
\def\separationVar{m}

\def\symbolDuration{T_{\rm s}}

\def\lengthGaGb{N}
\def\lengthGcGd{M}

\def\apac[#1][#2]{\rho_{#1}(#2)}
\def\apacPositive[#1][#2]{\rho^{+}_{#1}(#2)}
\def\lagForCorrelation{k}

\def\elex[#1]{{x}_{#1}}
\def\eles[#1]{{s}_{#1}}
\def\eley[#1]{{y}_{#1}}

\def\complexNumbers{\mathbb{C}}

\def\integers{\mathbb{Z}}
\def\integersPositive{\mathbb{Z}^{+}}

\def\flipConjugate[#1]{{{\tilde{#1}}}}

\def\polySeq[#1][#2]{p_{#1}(#2)}
\def\polyVariable{z}

\def\constante{{\rm e}}
\def\constanti{{\rm i}}
\def\constantMinusi{{\rm j}}
\def\constantOne{{\rm +}}
\def\constantMinusOne{{\rm -}}

\def\timeVar{t}

\def\angleGolay[#1]{\omega_{#1}}
\def\separationGolay[#1]{d_{#1}}

\def\exponentialBase{\xi}

\def\upsampleOp[#1][#2]{{\uparrow_{#1}\{#2\}}}

\newcommand\mydots{\hbox to 1em{.\hss.\hss.}}

\documentclass[journal]{IEEEtran}

\usepackage{multicol}
\usepackage{lipsum}
\usepackage{mathtools}
\usepackage{amsthm}
\usepackage{acronym}
\usepackage{stfloats}
\usepackage{xcolor}
\usepackage{amsfonts}
\usepackage{acronym}
\usepackage{cite}
\usepackage{bm}
\usepackage{amsmath,amssymb}
\usepackage{graphicx}
\usepackage[english]{babel}
\usepackage[caption=false,font=footnotesize,skip=0pt]{subfig}
\usepackage[geometry]{ifsym}
\usepackage{array}
\usepackage[utf8]{inputenc}
\usepackage[T1]{fontenc}
\usepackage{array}
\usepackage{algorithm}
\usepackage{algpseudocode}

\makeatletter
\def\BState{\State\hskip-\ALG@thistlm}
\makeatother
\newcolumntype{L}[1]{>{\raggedright\let\newline\\\arraybackslash\hspace{0pt}}m{#1}}
\newcolumntype{C}[1]{>{\centering\let\newline\\\arraybackslash\hspace{0pt}}m{#1}}
\newcolumntype{R}[1]{>{\raggedleft\let\newline\\\arraybackslash\hspace{0pt}}m{#1}}

\usepackage{cuted}
\makeatletter
\newif\ifAC@uppercase@first%
\def\Aclp#1{\AC@uppercase@firsttrue\aclp{#1}\AC@uppercase@firstfalse}%
\def\AC@aclp#1{%
	\ifcsname fn@#1@PL\endcsname%
	\ifAC@uppercase@first%
	\expandafter\expandafter\expandafter\MakeUppercase\csname fn@#1@PL\endcsname%
	\else%
	\csname fn@#1@PL\endcsname%
	\fi%
	\else%
	\AC@acl{#1}s%
	\fi%
}%
\def\Acp#1{\AC@uppercase@firsttrue\acp{#1}\AC@uppercase@firstfalse}%
\def\AC@acp#1{%
	\ifcsname fn@#1@PL\endcsname%
	\ifAC@uppercase@first%
	\expandafter\expandafter\expandafter\MakeUppercase\csname fn@#1@PL\endcsname%
	\else%
	\csname fn@#1@PL\endcsname%
	\fi%
	\else%
	\AC@ac{#1}s%
	\fi%
}%
\def\Acfp#1{\AC@uppercase@firsttrue\acfp{#1}\AC@uppercase@firstfalse}%
\def\AC@acfp#1{%
	\ifcsname fn@#1@PL\endcsname%
	\ifAC@uppercase@first%
	\expandafter\expandafter\expandafter\MakeUppercase\csname fn@#1@PL\endcsname%
	\else%
	\csname fn@#1@PL\endcsname%
	\fi%
	\else%
	\AC@acf{#1}s%
	\fi%
}%
\def\Acsp#1{\AC@uppercase@firsttrue\acsp{#1}\AC@uppercase@firstfalse}%
\def\AC@acsp#1{%
	\ifcsname fn@#1@PL\endcsname%
	\ifAC@uppercase@first%
	\expandafter\expandafter\expandafter\MakeUppercase\csname fn@#1@PL\endcsname%
	\else%
	\csname fn@#1@PL\endcsname%
	\fi%
	\else%
	\AC@acs{#1}s%
	\fi%
}%
\edef\AC@uppercase@write{\string\ifAC@uppercase@first\string\expandafter\string\MakeUppercase\string\fi\space}%
\def\AC@acrodef#1[#2]#3{%
	\@bsphack%
	\protected@write\@auxout{}{%
		\string\newacro{#1}[#2]{\AC@uppercase@write #3}%
	}\@esphack%
}%
\def\Acl#1{\AC@uppercase@firsttrue\acl{#1}\AC@uppercase@firstfalse}
\def\Acf#1{\AC@uppercase@firsttrue\acf{#1}\AC@uppercase@firstfalse}
\def\Ac#1{\AC@uppercase@firsttrue\ac{#1}\AC@uppercase@firstfalse}
\def\Acs#1{\AC@uppercase@firsttrue\acs{#1}\AC@uppercase@firstfalse}

\newtheorem{theorem}{Theorem}

\acrodef{PAPR}{peak-to-average-power ratio}
\acrodef{APAC}{aperiodic auto correlation}
\acrodef{OFDM}{orthogonal frequency division multiplexing}
\acrodef{DFT}{discrete Fourier transform}
\acrodef{IDFT}{inverse discrete Fourier transform}
\acrodef{DC}{direct current}
\acrodef{CS}{complementary sequence}
\acrodef{GCP}{Golay complementary pair}
\acrodef{ANF}{algebraic normal form}
\acrodef{PSK}{phase shift keying}
\acrodef{QAM}{quadrature amplitude modulation}
\acrodef{QPSK}{quadrature phase shift keying}
\acrodef{GDJ}{Golay-Davis-Jedwab}
\acrodef{PMEPR}{peak-to-mean envelope power ratios}
\acrodef{FFT}{fast Fourier transform}
\acrodef{BER}{bit-error rate}
\acrodef{SNR}{signal-to-noise ratio}
\acrodef{4G}{Fourth Generation}
\acrodef{5G}{Fifth Generation}
\acrodef{NR}{New Radio}
\acrodef{LTE}{Long-Term Evolution}
\acrodef{PTS}{partial transmit sequences}
\acrodef{PSD}{power spectral density}
\acrodef{LDPC}{low-density parity check}
\acrodef{OCB}{occupied channel bandwidth}
\acrodef{CP}{cyclic prefix}
\acrodef{CM}{cubic metric}
\acrodef{UE}{user equipment}
\acrodef{DAC}{digital-to-analog converter}
\acrodef{RS}{reference symbol}
\acrodef{LAA}{licensed-assisted access}
\acrodef{eLAA}{enhanced licensed-assisted access}
\acrodef{RB}{resource block}
\acrodef{OCB}{occupied channel bandwidth}
\acrodef{NR-U}{NR-Unlicensed}
\acrodef{IMD}{inter-modulation distortion}
\acrodef{ZC}{Zadoff-Chu}
\acrodef{SR}{scheduling request}
\acrodef{i.i.d}{independent-and-identically distributed}
\acrodef{NC}{non-contiguous}

\begin{document}
\title{ 
A Reliable Uplink Control Channel Design with Complementary Sequences
}
\author{Alphan~\c{S}ahin,~\IEEEmembership{Member,~IEEE,}
        and~Rui~Yang,~\IEEEmembership{Member,~IEEE}
\thanks{Alphan~\c{S}ahin and Rui~Yang are affiliated with InterDigital, Huntington Quadrangle, Melville, NY. email: \{alphan.sahin, rui.yang\}@interdigital.com}}
\maketitle

\begin{abstract}
In this study, 
we propose two schemes for uplink control channels based on non-contiguous \acp{CS} where the \ac{PAPR} of the resulting \ac{OFDM} signal is always less than or equal to $3$~dB. 
To obtain the proposed schemes, we extend Golay's concatenation and interleaving methods by considering extra upsampling and shifting parameters. The proposed schemes enable a flexible non-contiguous resource allocation in frequency, e.g., an arbitrary number of null symbols between the occupied \acp{RB}. 
The first scheme separates the \ac{PAPR} minimization and the inter-cell interference minimization problems. While the former is solved by spreading the sequences in a \ac{GCP} with the sequences in another \ac{GCP}, the latter is managed by designing a set of \acp{GCP} with low cross-correlation. 
The second scheme generates \acp{RS} and data symbols on each \ac{RB} as parts of an encoded \ac{CS}.
Therefore, it enables coherent detection at the receiver side.
The numerical results show that the proposed schemes offer significantly improved \ac{PAPR} and \ac{CM} results in case of non-contiguous resource allocation as compared to the sequences defined in 3GPP \ac{NR} and \ac{ZC} sequences.
\end{abstract}


\acresetall

\section{Introduction}

Non-contiguous resource allocation in the frequency domain is a well-known method to enhance the reliability of a link via frequency diversity gain. However, it is often  demoted or left as an optional feature  as it can cause \ac{IMD} products located outside of the bandwidth; and therefore may violate the emission requirements. On the other hand, for unlicensed bands, non-contiguous resource allocation is considered as a baseline  in today's major standards such as 3GPP \ac{LTE} \ac{eLAA}, MulteFire, and 3GPP \ac{NR-U}. 
The main reason behind the non-contiguous  resource  allocation is that it enables multiple accessing  in the uplink while allowing a radio to increase the transmit power under stringent \ac{PSD} and \ac{OCB} requirements imposed by the regulatory agencies.
For example, based on ETSI regulations \cite{etsi_2017}, the \ac{PSD} of the transmitted signal should be less than $10$~dBm/Mhz while the \ac{OCB} should be larger than $80\%$ of the nominal channel bandwidth in the $5$~GHz band. Therefore, the maximum transmit power of an uplink signal which consists of only a single \ac{RB} (e.g., $180$~kHz in \ac{LTE}), will be limited to $10$~dBm and the narrow bandwidth transmission will violate the regulations due to the \ac{OCB} requirement. 
To be able to increase the transmit power under the \ac{PSD} constraint, while complying with the \ac{OCB} requirement, non-contiguous resource allocation is adopted in \ac{LTE} \ac{eLAA} uplink, which is a major difference as compared to the one for legacy \ac{LTE}. The basic unit of the resource allocation for \ac{LTE} \ac{eLAA} data channels is defined as an {\em interlace} which is composed of $10$ equally-spaced \acp{RB} within a 20 MHz  bandwidth. A similar non-contiguous allocation, but more flexible in terms of bandwidth and subcarrier spacing, is also expected to be considered in \ac{NR-U}.

The instantaneous peak power of \ac{OFDM} signal with arbitrary information symbols in frequency can be high, which can degrade the transmission power efficiency and decrease the coverage range of a link due to the power back-off. Non-contiguous resource allocation introduces an additional challenging constraint on \ac{PAPR} minimization.  This issue can also be a detrimental factor for the reliability, particularly when transmitting very short packet with one or two \ac{OFDM} symbols for latency reduction. In this study, we address the issue of the high instantaneous peak power of an \ac{OFDM} symbol with a non-contiguous resource allocation and consider the cases where the transmitter needs to transmit a small amount of information such as ACK/NACK or \ac{SR} in the uplink.

In the literature, there are many approaches  investigating \ac{PAPR} minimization for \ac{OFDM} \cite{Rahmatallah_2013}. For example, with \ac{PTS} \cite{Muller_1997}, additional phase rotations are applied to the symbol groups in frequency such that the resulting signal has low \ac{PAPR}. 
However, \ac{PTS} can increase the overhead as the receiver may need to know the rotations. 
Companding transform is another widely-used method which compensates the distortion from hardware non-linearity at the expense of higher \ac{BER} \cite{Huang_2004}. Another approach is \ac{DFT} precoding\cite{Sari_1995}, i.e., \ac{DFT}-spread \ac{OFDM}, where the multicarrier structure of a plain \ac{OFDM} symbol is effectively converted to a wideband single carrier waveform \cite{sahin_2018icc}. It substantially decreases the fluctuations in time when the resource allocation is contiguous in the frequency and low-order modulation symbols are utilized. It also allows several methods  such as frequency domain windowing to decrease the \ac{PAPR} further \cite{Sahin_2016}. On the other hand,  in cases of non-contiguous resource allocation after \ac{DFT} precoding, the low-\ac{PAPR} benefit of \ac{DFT}-spread \ac{OFDM} diminishes as it loses its single carrier structure. 
The approaches that take the encoding into account for reducing \ac{PAPR} may require a joint design that typically imposes additional constraints on coding structure, modulation type, and waveform parameters such as resource allocation. For example, by exhaustive search for parity bits  that lead to low \ac{PAPR} \cite{Jones_1994} or by using offsets from linear code \cite{Jones_1996} are several methods that place constraints on the parity bits. In \cite{Daoud_2009}, Daoud and Alani proposed to use \ac{LDPC} codes to mitigate the \ac{PAPR} of \ac{OFDM} symbols via exhaustive search. In \cite{Tsai_2008}, various interleavers were proposed for Turbo encoder to reduce \ac{PAPR}.
To mitigate \ac{PAPR} via encoding, a remarkable method has been established with \acp{CS} \cite{Golay_1961},  especially after the connection between \acp{CS} and Reed-Muller codes was discovered by Davis and Jedwab \cite{davis_1999}. However, synthesizing \acp{CS} for a given resource allocation is still a challenging task. Recently, a theoretical framework was proposed to synthesize a \ac{CS} with null symbols, i.e., non-contiguous \ac{CS} \cite{Sahin_2018}. Nevertheless, the practical applications of non-contiguous \acp{CS} are still in their early stage.

In this study, we propose two schemes for the uplink control channel where the \ac{PAPR} of the resulting \ac{OFDM} symbol is restricted below a certain level by exploiting the non-contiguous \acp{CS} obtained via Theorem \ref{th:golayIterative} given in Section \ref{sec:csPucch}. The first scheme enables non-coherent detection at the receiver while mitigating the interference from other cells in the network with well-designed \acp{CS} that restrict the \ac{PAPR} to be less than or equal to $3$ dB. 
This scheme can be considered as an extension of the uplink control channel Format 0 in 3GPP 
\ac{NR} developed for licensed bands \cite{nr_phy_2017}. 
The second scheme exploits the properties of Theorem \ref{th:golayIterative} and yields an \ac{OFDM} symbol which includes built-in reference symbols as part of encoded \ac{CS}. Therefore, it enables coherent detection at the receiver side.

The rest of the paper is organized as follows. In Section \ref{sec:prelim}, we provide preliminary discussions on the polynomial representation of sequences and \acp{GCP}. In Section \ref{sec:csPucch}, we provide Theorem \ref{th:golayIterative} and discuss the proposed schemes. In Section \ref{sec:numerical}, we present numerical results and compare them with other potential approaches. We conclude the paper in Section \ref{sec:conclusion}.

{\em Notation:} The field of complex numbers,  the set of integers, and the set of positive integers are denoted by $\complexNumbers$, $\integers$, and $\integersPositive$, respectively. 
The symbols $\constanti$, $\constantMinusi$, $\constantOne$, and $\constantMinusOne$ denote $\sqrt{-1}$, $-\sqrt{-1}$, $1$, and $-1$, respectively. 
A sequence of length $\lengthGaGb$ is represented by $\seqGa= (\eleGa[0],\eleGa[1],\dots, \eleGa[\lengthGaGb-1])$. The element-wise complex conjugation and the element-wise absolute operation are denoted by  $(\cdot)^*$ and $|\cdot|$, respectively. 
The operator $\flipConjugate[\seqGa]$ reverses the order of the elements of $\seqGa$ and applies element-wise complex conjugation.
The operation $\upsampleOp[\upsampleVarA][\seqGa]$ introduces $\upsampleVarA-1$ null symbols between the elements of $\seqGa$.
The operations $\seqGa\pm\seqGb$,  $\seqGa \odot \seqGb$, $\seqGa*\seqGb$, and $\langle\seqGa,\seqGb\rangle$ are the element-wise summation/subtraction, the element-wise multiplication, linear convolution, and the inner product of $\seqGa$ and $\seqGb$, respectively.

\section{Preliminaries and Further Notation}
\label{sec:prelim}
\subsection{Polynomial Representation of a Sequence}
\label{subsec:poly}
The polynomial representation of the sequence $\seqGa$ can be given by
\begin{align}
\polySeq[\seqGa][\polyVariable] \triangleq \eleGa[\lengthGaGb-1]\polyVariable^{\lengthGaGb-1} + \eleGa[\lengthGaGb-2]\polyVariable^{\lengthGaGb-2}+ \dots + \eleGa[0]~,
\label{eq:polySeq}
\end{align} 
where $\polyVariable\in \complexNumbers$ is a complex number.  One can show that the polynomial $\polySeq[\seqGa][\polyVariable^\upsampleVarA]$,  $\polySeq[\seqGa][\polyVariable^\upsampleVarA]\polySeq[\seqGb][\polyVariable^\upsampleVarB]$, and $\polySeq[\seqGa][\polyVariable]\polyVariable^\separationVar$ represent the up-sampled sequence $\seqGa$ with the factor of $\upsampleVarA\in\integersPositive$,
the convolution of the up-sampled sequence $\seqGa$ with the factor of $\upsampleVarA$ and the up-sampled sequence $\seqGb$ with the factor of $\upsampleVarB\in\integersPositive$, and   the sequence $\seqGa$ padded with $\separationVar\in\integersPositive$ null symbols, respectively. In addition, by restricting $\polyVariable$ to be on the unit circle in the complex plane, i.e., $\polyVariable\in\{\constante^{\constanti\frac{2\pi\timeVar}{\symbolDuration}}| 0\le\timeVar <\symbolDuration
\}$, the polynomial representation given in \eqref{eq:polySeq} corresponds to an \ac{OFDM} symbol in continuous time  where the elements of the sequence $\seqGa$ are mapped to the subcarriers with the same order and $\symbolDuration$ denotes the \ac{OFDM} symbol duration.

\subsection{Golay Complementary Pair and Complementary Sequence}
The sequence pair  $(\seqGa,\seqGb)$ of length  $\lengthGaGb$ is called a \ac{GCP} if
\begin{align}
\apac[\seqGa][\lagForCorrelation]+\apac[\seqGb][\lagForCorrelation] = 0,~~ \text{for}~ \lagForCorrelation\neq0~
\label{eq:GCP}
\end{align}
where  $\apac[\seqGa][\lagForCorrelation]$   is the \ac{APAC} of the sequence $\seqGa$ given by
\begin{align}
\apac[\seqGa][\lagForCorrelation] \triangleq 
\begin{cases}
\apacPositive[\seqGa][\lagForCorrelation], 		& \text{if}~ \lagForCorrelation \ge 0\\
\apacPositive[\seqGa][-\lagForCorrelation]^*,     	&\text{if}~ \lagForCorrelation < 0
\end{cases}~,
\nonumber
\end{align}
and $\apacPositive[\seqGa][\lagForCorrelation]=\sum_{\indexEleOfSeq=0}^{\lengthGaGb-\lagForCorrelation-1} \eleGa[\indexEleOfSeq]^*\eleGa[\indexEleOfSeq+\lagForCorrelation] $ for
$\lagForCorrelation = 0,1,\dots,\lengthGaGb-1$ and $0$ otherwise.  The sequences $\seqGa$ and $\seqGb$ are defined as \acp{CS}. By using the definition of \ac{GCP}, one can show that the GCP $(\seqGa,\seqGb)$ satisfies 
\begin{align}
\polySeq[\seqGa][\polyVariable]\polySeq[\seqGa^*][\polyVariable^{-1}]+\polySeq[\seqGb][\polyVariable]\polySeq[\seqGb^*][\polyVariable^{-1}] =\apac[\seqGa][0]+\apac[\seqGb][0]~.
\label{eq:zDomainGCP}
\end{align}
By restricting $\polyVariable$ to be on the unit circle as a further condition, \eqref{eq:zDomainGCP} can be written as
\begin{align}
|\polySeq[\seqGa][\polyVariable]|^2+|\polySeq[\seqGb][\polyVariable]|^2\bigg\rvert_{\polyVariable=\constante^{\constanti\frac{2\pi\timeVar}{\symbolDuration}}} =\underbrace{\apac[\seqGa][0]+\apac[\seqGb][0]}_{\text{constant}}~.
\label{eq:timeDomainGCP}
\end{align}

 The main property that we inherited from \acp{GCP} in this study is that the instantaneous peak power of the corresponding \ac{OFDM} signal generated through a \ac{CS} $\seqGa$ is bounded, i.e., $\max|\polySeq[\seqGa][\polyVariable]|^2 \le \apac[\seqGa][0]+\apac[\seqGb][0]$. Therefore, based on \eqref{eq:timeDomainGCP}, the \ac{PAPR} of the \ac{OFDM} signal is less than or equal to $10\log_{10}(2)\approx3$~dB if $\apac[\seqGa][0]=\apac[\seqGb][0]$ \cite{boyd_1986}. For the other interesting properties of \acp{GCP}, we refer the reader to an excellent survey given in \cite{parker_2003}.

\section{Complementary Sequence-Based Interlace Design}
\label{sec:csPucch}

\begin{figure}[t]
	\centering
 	{\includegraphics[width =3.0in]{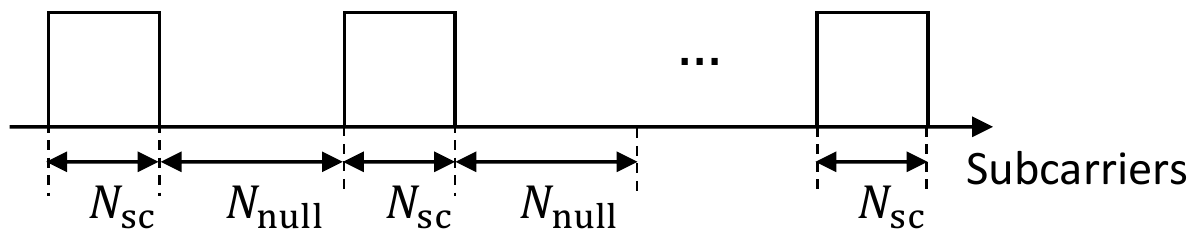}
 	}
 	\caption{Interlace model.}
 	\label{fig:interlace}
 \end{figure} 
We model an interlace as a non-contiguous resource allocation which consists of $\numberOfClusters$ \acp{RB} each of which composed of $\RBsize$ subcarriers where the \acp{RB} are separated by $\numberOfNulls$ tones in the frequency domain as shown in \figurename~\ref{fig:interlace}. Based on this interlace model, we consider two schemes where the first one primarily targets non-coherent detectors at receiver and the second scheme enables coherent detectors by yielding to \acp{RS} in each \ac{RB}. To explain the origin of the proposed schemes, we first restate the following theorem which generalizes Golay's concatenation and interleaving methods \cite{Golay_1961,parker_2003}:
\begin{theorem}
	\label{th:golayIterative}
Let $(\seqGa,\seqGb)$ and $(\seqGc,\seqGd)$ be \acp{GCP} of length $\lengthGaGb$ and $\lengthGcGd$, respectively, and $\angleGolay[1],\angleGolay[2]\in \{u:u\in\complexNumbers, |u|=1\}$  and $\upsampleVarA, \upsampleVarB, \separationVar\in\integers$. Then, the sequences $\seqGf$ and $\seqGg$ where their polynomial representations given by
\begin{align}
\polySeq[\seqGf][\polyVariable] =& \angleGolay[1]\polySeq[{\seqGa}][\polyVariable^\upsampleVarA]\polySeq[{\seqGc}][\polyVariable^\upsampleVarB]   
+ \angleGolay[2]\polySeq[{\seqGb}][\polyVariable^\upsampleVarA]\polySeq[{\seqGd}][\polyVariable^\upsampleVarB]  \polyVariable^{\separationVar} ~, 
\label{eq:gcpGf}
\\
\polySeq[\seqGg][\polyVariable] =& \angleGolay[1]\polySeq[{\seqGa}][\polyVariable^\upsampleVarA]\polySeq[{\flipConjugate[\seqGd]}][\polyVariable^\upsampleVarB] 
- \angleGolay[2]\polySeq[{\seqGb}][\polyVariable^\upsampleVarA]\polySeq[{\flipConjugate[\seqGc]}][\polyVariable^\upsampleVarB]  \polyVariable^{\separationVar} ~, 
\label{eq:gcpGg}
\end{align}
construct a \ac{GCP}.
\end{theorem}

\begin{proof}
Since the sequence pairs $(\seqGa,\seqGb)$ and $(\seqGc,\seqGd)$ are \acp{GCP}, by the definition, $|\polySeq[{\seqGa}][\polyVariable]|^2 + |\polySeq[{\seqGb}][\polyVariable]|^2=C_1$ and $|\polySeq[{\seqGc}][\polyVariable]|^2 + |\polySeq[{\seqGd}][\polyVariable]|^2=C_2$, where $C_1$ and $C_2$ are some constants.
To prove that the sequences $\seqGf$ and $\seqGg$ generated through \eqref{eq:gcpGf} and \eqref{eq:gcpGg} construct a GCP, we  need to show that $|\polySeq[{\seqGf}][\polyVariable]|^2+|\polySeq[{\seqGg}][\polyVariable]|^2$ is also a constant. 
 By exploiting the fact the polynomial representation of the sequence $\flipConjugate[\seqGa]$ can be calculated as $\polySeq[{\flipConjugate[\seqGa]}][\polyVariable^\upsampleVarA]=\polySeq[{\seqGa^*}][\polyVariable^{-\upsampleVarA}]\polyVariable^{\upsampleVarA\lengthGaGb-\upsampleVarA}$, one can calculate $|\polySeq[{\seqGf}][\polyVariable]|^2+|\polySeq[{\seqGg}][\polyVariable]|^2$ as
\begin{align}
|\polySeq[{\seqGf}][\polyVariable]&|^2+|\polySeq[{\seqGg}][\polyVariable]|^2 \nonumber
\\=
&~~~
(\angleGolay[1]\polySeq[{\seqGa}][\polyVariable^\upsampleVarA]\polySeq[{\seqGc}][\polyVariable^\upsampleVarB] 
+ \angleGolay[2]\polySeq[{\seqGb}][\polyVariable^\upsampleVarA]\polySeq[{\seqGd}][\polyVariable^\upsampleVarB]  \polyVariable^{\separationVar})
\nonumber\\
&\times 
(\angleGolay[1]^*\polySeq[{\seqGa^*}][\polyVariable^{-\upsampleVarA}]\polySeq[{\seqGc^*}][\polyVariable^{-\upsampleVarB}]   
+ \angleGolay[2]^*\polySeq[{\seqGb^*}][\polyVariable^{-\upsampleVarA}]\polySeq[{\seqGd^*}][\polyVariable^{-\upsampleVarB}]  \polyVariable^{-\separationVar})\nonumber\\
&+
(\angleGolay[1]\polySeq[{\seqGa}][\polyVariable^\upsampleVarA]\polySeq[{\flipConjugate[\seqGd]}][\polyVariable^\upsampleVarB] 
- \angleGolay[2]\polySeq[{\seqGb}][\polyVariable^\upsampleVarA]\polySeq[{\flipConjugate[\seqGc]}][\polyVariable^\upsampleVarB]  \polyVariable^{\separationVar} )
\nonumber\\&\times
(\angleGolay[1]^*\polySeq[{\seqGa^*}][\polyVariable^{-\upsampleVarA}]\polySeq[{\flipConjugate[\seqGd^*]}][\polyVariable^{-\upsampleVarB}] 
- \angleGolay[2]^*\polySeq[{\seqGb^*}][\polyVariable^\upsampleVarA]\polySeq[{\flipConjugate[\seqGc]^*}][\polyVariable^{-\upsampleVarB}]  \polyVariable^{-\separationVar} )
\nonumber \\
\stackrel{(a)}{=}&~~~ 
\polySeq[{\seqGa}][\polyVariable^\upsampleVarA]\polySeq[{\seqGa^*}][\polyVariable^{-\upsampleVarA}]
\polySeq[{\seqGc}][\polyVariable^\upsampleVarB]\polySeq[{\seqGc^*}][\polyVariable^{-\upsampleVarB}]
\nonumber \\
&+\polySeq[{\seqGa}][\polyVariable^\upsampleVarA]\polySeq[{\seqGa^*}][\polyVariable^{-\upsampleVarA}]
\polySeq[\flipConjugate[{\seqGd}]][\polyVariable^\upsampleVarB]\polySeq[\flipConjugate[{\seqGd}]^*][\polyVariable^{-\upsampleVarB}]
\nonumber \\
&+\polySeq[{\seqGb}][\polyVariable^\upsampleVarA]\polySeq[{\seqGb^*}][\polyVariable^{-\upsampleVarA}]
\polySeq[\flipConjugate[{\seqGc}]][\polyVariable^\upsampleVarB]\polySeq[\flipConjugate[{\seqGc}]^*][\polyVariable^{-\upsampleVarB}]
\nonumber \\
&+\polySeq[{\seqGb}][\polyVariable^\upsampleVarA]\polySeq[{\seqGb^*}][\polyVariable^{-\upsampleVarA}]
\polySeq[{\seqGd}][\polyVariable^\upsampleVarB]\polySeq[{\seqGd^*}][\polyVariable^{-\upsampleVarB}]
\nonumber\\
\stackrel{(b)}{=}&~~~ 
(\polySeq[{\seqGa}][\polyVariable^\upsampleVarA]\polySeq[{\seqGa^*}][\polyVariable^{-\upsampleVarA}]
+\polySeq[{\seqGb}][\polyVariable^\upsampleVarA]\polySeq[{\seqGb^*}][\polyVariable^{-\upsampleVarA}])
\nonumber \\
&\times (\polySeq[{\seqGc}][\polyVariable^\upsampleVarB]\polySeq[{\seqGc^*}][\polyVariable^{-\upsampleVarB}]+
\polySeq[{\seqGd}][\polyVariable^\upsampleVarB]\polySeq[{\seqGd^*}][\polyVariable^{-\upsampleVarB}])
\nonumber\\
\stackrel{(c)}{=}&~~~ C_1C_2
\end{align}
where (a) follows from $\polySeq[\flipConjugate[{\seqGc}]^*][\polyVariable^{-\upsampleVarB}]\polySeq[{\flipConjugate[\seqGd]}][\polyVariable^{\upsampleVarB}]=\polySeq[{\seqGc}][\polyVariable^{\upsampleVarB}]\polySeq[{\seqGd^*}][\polyVariable^{-\upsampleVarB}]$ and $\polySeq[\flipConjugate[{\seqGc}]][\polyVariable^{\upsampleVarB}]\polySeq[{\flipConjugate[\seqGd]^*}][\polyVariable^{-\upsampleVarB}]=\polySeq[{\seqGc^*}][\polyVariable^{-\upsampleVarB}]\polySeq[{\seqGd}][\polyVariable^{\upsampleVarB}]$, (b) is because  $\polySeq[\flipConjugate[{\seqGc}]][\polyVariable^\upsampleVarB]\polySeq[\flipConjugate[{\seqGc}]^*][\polyVariable^{-\upsampleVarB}]=\polySeq[{\seqGc^*}][\polyVariable^{-\upsampleVarB}]\polySeq[{\seqGc}][\polyVariable^\upsampleVarB]$ and $\polySeq[\flipConjugate[{\seqGd}]][\polyVariable^\upsampleVarB]\polySeq[\flipConjugate[{\seqGd}]^*][\polyVariable^{-\upsampleVarB}] = \polySeq[{\seqGd^*}][\polyVariable^{-\upsampleVarB}]\polySeq[{\seqGd}][\polyVariable^\upsampleVarB]
$ and (c) is because of the definition of a \ac{GCP}.
\end{proof}

Theorem \ref{th:golayIterative} is practically appealing since it can generate \acp{CS} with null symbols through the careful choice of the parameters $\separationVar$, $\upsampleVarA$, and $\upsampleVarB$  by starting from two \acp{GCP}. Therefore, it can be utilized for generating sequences for a non-contiguous resource allocation. In the following subsections, we use the relationships given in Section \ref{subsec:poly} for the polynomials and exploit Theorem \ref{th:golayIterative} to construct an interlace with the desired resource allocation in the frequency domain.

\subsection{Non-coherent Scheme}

In this scheme, we consider a unimodular sequence, i.e., a sequence where the amplitude of each element is $1$, for each \ac{RB} in an interlace. Unimodular sequences are suitable for an uplink control channel since  certain cyclic-shifts of an \ac{OFDM} signal generated through a unimodular sequence are orthogonal to each other \cite{Benedetto_2009}, which enables code-domain multiple access. Let $\seqGx=(\elex[0],\elex[1],\dots, \elex[\RBsize-1])$ be a unimodular base sequence, i.e., $|\elex[i]|=1$ for $i\in0,1,,\dots,\RBsize-1$. Then, by representing the cyclic-shift in time domain as modulation in frequency domain,
\begin{align}
\langle
\seqGyShift[{\shift[1]}],\seqGyShift[{\shift[2]}]\rangle = 0 ~ \text{if}~ \shift[1] \neq \shift[2]~,
\end{align}
where
\begin{align}
\seqGyShift[{\shift[]}] = \seqGx \odot \seqGs 
\label{eq:shifting}
\end{align}
for $\seqGs = {(\exponentialBase^{0\times\shift[]},\exponentialBase^{1\times\shift[]}, \dots, \exponentialBase^{(\RBsize-1)\times\shift[]})}$,  $\exponentialBase=\constante^{\constanti\frac{2\pi}{\RBsize}}$ and $\shift[1],\shift[2],\shift[]\in\integers$. When unimodular sequences are employed on each \ac{RB}, the number of orthogonal resources generated through cyclic-shifts in time is limited to the size of \acp{RB}. For example, if there are $\numberOfClusters=10$ \acp{RB} in one interlace and each \ac{RB} consists of $\RBsize=12$ subcarriers, there are $12$ orthogonal resources generated through the shifts in {\em time}. 12 resources can be shared by 6 users to transmit 1-bit information (e.g., ACK/NACK) or 3 users to transmit 2-bit information (e.g., ACK/NACK and scheduling request). The receiver can detect the corresponding sequences in each \ac{RB} non-coherently  to decode the information. 

\def\setC{C}
\def\setD{D}
\def\setQ{Q_1}
\def\setQr{Q_2}
\def\numberOfSequences{K}
\def\indexSequence{i}

\begin{figure*}[t]
	\centering
	\subfloat[Transmitter structure for non-coherent scheme.]{\includegraphics[width =2.65in]{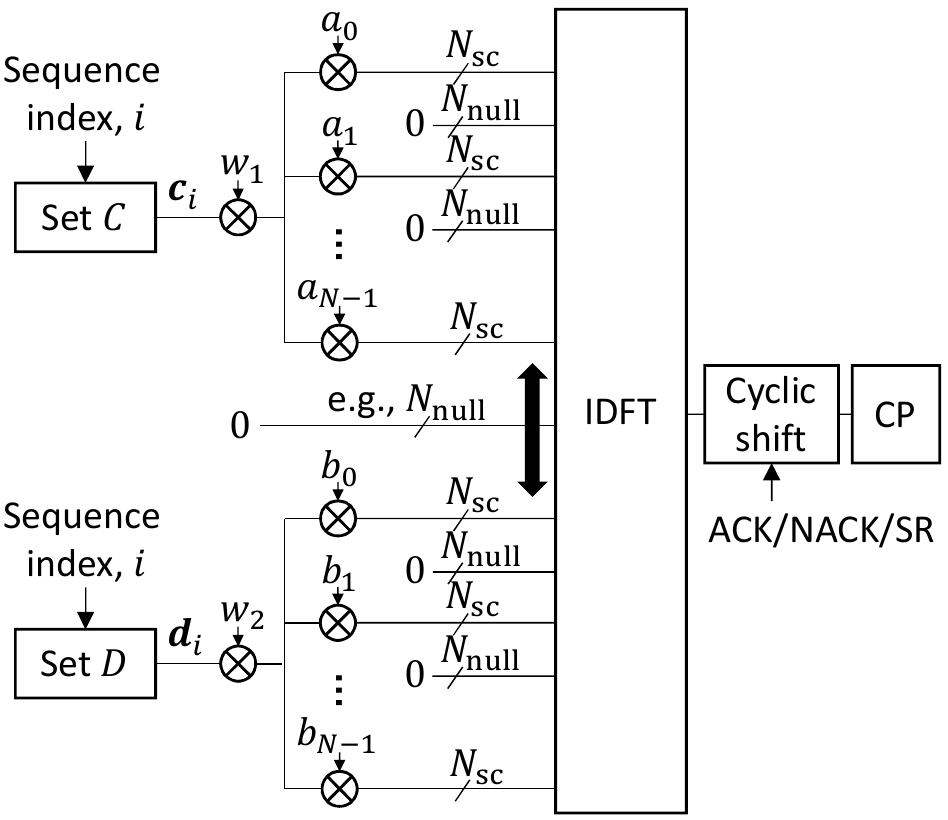}
	\label{fig:nco}}~~~~~~~~~~~~~~~~~~~
	\subfloat[Transmitter structure for coherent scheme.]{\includegraphics[width =2.0in]{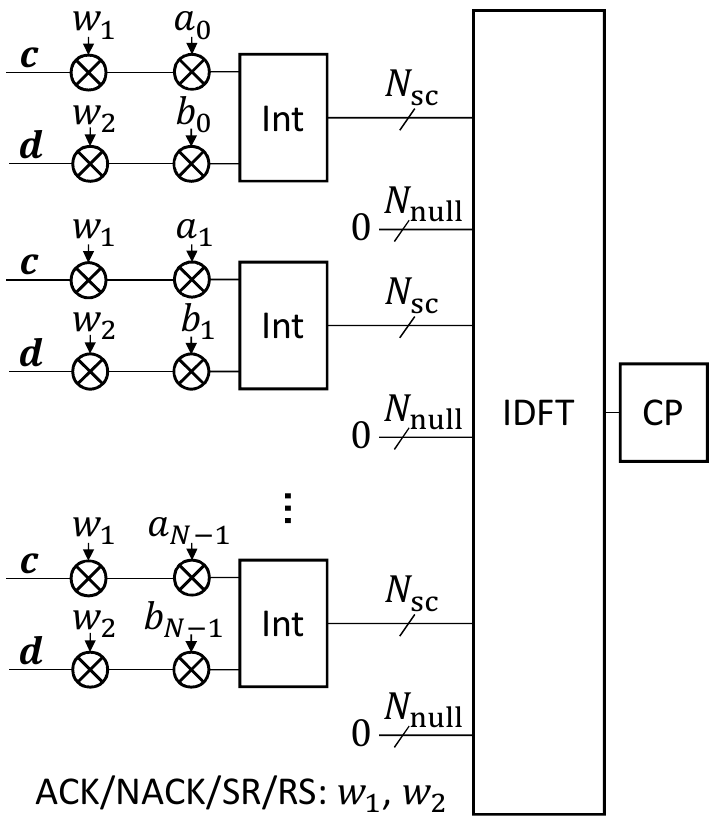}
		\label{fig:co}}
	\caption{Proposed transmitter structures based on \acp{CS}}
	\label{fig:tx}
\end{figure*}


To construct the interlace with the proposed scheme, we first choose a \ac{GCP} $(\seqGa,\seqGb)$ of length $\numberOfClusters/2$  and a \ac{GCP} $(\seqGc_\indexSequence,\seqGd_\indexSequence)$ of length $\RBsize$  where the elements of $\seqGa$, $\seqGb$, $\seqGc_\indexSequence$ and $\seqGd_\indexSequence$ are in the set $\setQ\triangleq\{\constantOne,\constantMinusOne,\constanti,\constantMinusi\}$
for $\indexSequence=1,2,\mydots,\numberOfSequences$ where $\numberOfSequences\in\integersPositive$.  We then generate the interlace through \eqref{eq:gcpGf} in Theorem \ref{th:golayIterative} by setting $\seqGc=\seqGc_\indexSequence$, $\seqGd=\seqGd_\indexSequence$, $\angleGolay[1]=\angleGolay[2]=\constante^{\frac{\constanti\pi}{4}}$, $\upsampleVarA = \RBsize+\numberOfNulls$, $\upsampleVarB =1$, and $\separationVar = (\RBsize+\numberOfNulls)\times\numberOfClusters/2$. 

In this scheme, $\seqGa$ and $\seqGb$ act as spreading sequences for the  sequences $\seqGc_\indexSequence$ and $\seqGd_\indexSequence$, respectively, since
\begin{align}
\polySeq[{\seqGa}][\polyVariable^{\RBsize+\numberOfNulls}]\polySeq[{\seqGc}][\polyVariable]=\polySeq[{\upsampleOp[{\RBsize+\numberOfNulls}][\seqGa]*\seqGc}][\polyVariable],
\end{align}
and
\begin{align}
\polySeq[{\seqGb}][\polyVariable^{\RBsize+\numberOfNulls}]\polySeq[{\seqGd}][\polyVariable]=\polySeq[{\upsampleOp[{\RBsize+\numberOfNulls}][\seqGb]*\seqGd}][\polyVariable].
\end{align}
In other words, the \acp{RB} are constructed with phase rotated versions of $\seqGc_\indexSequence$ and $\seqGd_\indexSequence$ and the phase rotations are determined by the elements of $\seqGa$ and $\seqGb$ as shown in \figurename~\ref{fig:tx}\subref{fig:nco}. Based on the second part of \eqref{eq:gcpGf} in Theorem \ref{th:golayIterative}, $\separationVar = (\RBsize+\numberOfNulls)\times\numberOfClusters/2$ null symbols are prepended to the sequence ${\upsampleOp[{\RBsize+\numberOfNulls}][\seqGb]*\seqGd_\indexSequence}$. Hence, while the first half of the interlace is a function of $\seqGa$ and $\seqGc_\indexSequence$, the second half of the interlace is generated through $\seqGb$ and $\seqGd_\indexSequence$ as illustrated in  \figurename~\ref{fig:tx}\subref{fig:nco}. 
By considering the interlace structure in \ac{LTE} \ac{eLAA}, i.e., $\RBsize=12$, $\numberOfClusters=10$,  $\numberOfNulls = 9\times12=108$ subcarriers, and \ac{QPSK} alphabet for the sequences, the interlace can be constructed with the proposed scheme via \eqref{eq:gcpGf} when $\upsampleVarA =120$, $\upsampleVarB =1$, and $\separationVar = 600$ and the sequences $\seqGa$, $\seqGb$, $\seqGc_\indexSequence$, $\seqGd_\indexSequence$ are chosen arbitrarily from the set provided from \cite{holzmann_1991} as $\seqGa = (\constantOne,\constantOne,\constantOne,\constantMinusi,\constanti)$, $\seqGb = (\constantOne,\constanti,\constantMinusOne,\constantOne,\constantMinusi)$, $\seqGc_\indexSequence =  (\constantOne,\constantOne,\constantOne,\constantOne,\constantMinusOne,\constantMinusOne,\constantMinusOne,\constantOne,\constanti,\constantMinusi,\constantMinusOne,\constantOne)$, and $\seqGd_\indexSequence = (\constantOne,\constantOne,\constanti,\constanti,\constantOne,\constantOne,\constantMinusOne,\constantOne,\constantOne,\constantMinusOne,\constantOne,\constantMinusOne)$.

The proposed scheme offers flexibility for the interlace structure since $\numberOfNulls$ and $\separationVar$ can be arbitrary values. For example, in one scenario, another user in the cell may need to transmit a random access signal by using a contiguous allocation\footnote{Contiguous allocation for random access signals typically improve the timing accuracy with a simple receiver.}. To address this scenario, using a larger $\separationVar$ and a {\em single} pair of shorter spreading sequences $\seqGa$ and $\seqGb$ can generate the desired gap in the frequency; and the \ac{PAPR} is still less than or equal to 3 dB for the same $\seqGc_\indexSequence$ and $\seqGd_\indexSequence$. A quick investigation also suggests that the interlace where $\seqGc_\indexSequence$ and $\seqGd_\indexSequence$ are on the adjacent \acp{RB} can be constructed when  $\upsampleVarA = 2(\RBsize+\numberOfNulls)$, $\upsampleVarB =1$, and $\separationVar = \RBsize+\numberOfNulls$, which is another design option offered by the proposed scheme. In this case, the interlace with \ac{LTE} parameters can be obtained for $\upsampleVarA = 240$, $\upsampleVarB =1$, and $\separationVar = 120$.

\def\thershold{\beta}
\def\seedSeqGcCandidate[#1]{{\seqGc_{#1}''}}
\def\seedSeqGdCandidate[#1]{{\seqGd_{#1}''}}
\def\seqGcCandidate[#1]{{\seqGc_{#1}'}}
\def\seqGdCandidate[#1]{{\seqGd_{#1}'}}
\def\upsampleValueU{u}
\def\seedA{S_c''}
\def\seedB{S_d''}
\def\seedAp{{S_c'}}
\def\seedBp{{S_d'}}
\def\seedSize{I}
\def\seedSizeEq{J}
\def\setS{S}
\def\indexForFirst{i}
\def\indexForSecond{j}
\def\correlationCalc{\alpha}

The proposed schemes simplify the control channel design as it separates two complicated problems, i.e., the \ac{PAPR} minimization for non-contiguous allocation and inter-cell interference minimization. For the first problem, as long as the sequences $\seqGc_\indexSequence$ and $\seqGd_\indexSequence$ for $\indexSequence=1,2,\mydots,\numberOfSequences$ construct a \ac{GCP}, the same spreading $\seqGa$ and $\seqGb$ can be utilized to limit the \ac{PAPR}. Note that \ac{NR-U} may be configured with multiple subcarrier spacing options, e.g., 15, 30, or 60 KHz, and bandwidth, e.g., 20, 40, or 80 MHz. For different configurations,  a single \ac{GCP} is able to limit \ac{PAPR}  to less than or equal to $3$~dB for all of the sequences in $\setC\triangleq\{\seqGc_1,\seqGc_2,\mydots,\seqGc_\numberOfSequences\}$ and $\setD\triangleq\{\seqGd_1,\seqGd_2,\mydots,\seqGd_\numberOfSequences\}$, which remarkably reduces the design complexity. For the second problem, the cross-correlation between the sequences used at different cells, i.e., $\seqGc_\indexSequence$ and $\seqGc_j$ for $\indexSequence\neq j$ (and  $\seqGd_\indexSequence$ and $\seqGd_j$), should be as low as possible to minimize the potential interference among the different cells.
Due to the imperfect timing alignment between the uplink signals and the multipath channel, the signal may also be exposed to additional shift in time within the \ac{CP}. In this case, the cross-correlation between the sequences in the sets should consider not only integer shifts, but also non-integer $\shift[]{}$ values in \eqref{eq:shifting}. In \ac{NR} and \ac{LTE}, the number of available base sequences  is set to $\numberOfSequences=30$ for $\RBsize = 12$. However,  designing $\setC$ and $\setD$ with a small peak cross-correlation value $\thershold$, i.e., $\langle\seqGc_\indexSequence,\seqGc_j\odot \seqGs \rangle\le\thershold$ and  $\langle\seqGd_\indexSequence,\seqGd_j\odot \seqGs \rangle\le\thershold$  for $\indexSequence\neq j$ and $\indexSequence,j\in \{1,2,\mydots,\numberOfSequences\}$ and $\shift[]{}\in[0,\RBsize-1]$ is challenging task. Therefore, re-using the sets for different configurations is highly desirable. 
This naturally leads to the following question for the proposed scheme: {\em Are there any $\setC$ and $\setD$ for $\numberOfSequences=30$ and $\RBsize=12$, i.e., 30 different \acp{GCP} of length 12, such that the peak cross-correlation between any two sequences in each set for any non-integer $\shift[]{}$ value that is sufficiently small?}

To answer this question, we consider a  procedure which exploits computer-generated \acp{GCP} provided in \cite{holzmann_1991} for length 12 to obtain $\setC$ and $\setD$. We initialize the algorithm with $\seedSize=52$ \acp{GCP} of length 12 listed in  \cite{holzmann_1991} and populate as $\seedA=\{\seedSeqGcCandidate[1],\mydots,\seedSeqGcCandidate[\seedSize]\}$ and $\seedB=\{\seedSeqGdCandidate[1],\mydots,\seedSeqGdCandidate[\seedSize]\}$. For the $\indexForFirst$th seed \ac{GCP} $({\seedSeqGcCandidate[\indexForFirst]},{\seedSeqGdCandidate[\indexForFirst]})$, we first enumerate $\seedSizeEq=8$ equivalent \acp{GCP} by interchanging, reflecting both (i.e., reversing the order of the elements of the sequences), and conjugate reflecting original sequences in the seed \ac{GCP}, which lead to the sets $\seedAp=\{\seqGcCandidate[1],\mydots,\seqGcCandidate[\seedSizeEq]\}$ and  $\seedBp=\{\seqGdCandidate[1],\mydots,\seqGdCandidate[\seedSizeEq]\})$. Because of the properties of \ac{GCP}, the $(\seqGcCandidate[\indexForSecond],\seqGdCandidate[\indexForSecond])$ still constructs \acp{GCP} for $\indexForSecond=1,\mydots,\seedSizeEq$. For a given candidate \ac{GCP} $(\seqGcCandidate[\indexForSecond],\seqGdCandidate[\indexForSecond])$, we calculate $\langle\seqGc_\indexSequence,\seqGcCandidate[\indexForSecond]\odot \seqGs \rangle$ and  $\langle\seqGd_\indexSequence,\seqGdCandidate[\indexForSecond]\odot \seqGs \rangle$  for $\seqGc_\indexSequence\in\setC$ and $\seqGd_\indexSequence\in\setD$ and  $\shift[]\in\{0,1/\lengthGcGd\upsampleValueU,\mydots,(\lengthGcGd\upsampleValueU-1)/\lengthGcGd\upsampleValueU\}$ and  $\upsampleValueU>1$. If the results are less than or equal to $\thershold$ for all $\shift[]$, we update $\setC$ and $\setD$ by including the sequences in the candidate \ac{GCP} to the sets. 

\def\seqCelep[#1][#2]{\phi_{#1#2}}
\def\seqDelep[#1][#2]{\theta_{#1#2}}

\begin{table}[t]
    \centering
	\caption{The sequences in $\setC$ and $\setD$}
	\label{table:seqC}
		\begin{tabular}{r|cc}
$\indexSequence$  & $\seqGc_\indexSequence$ & $\seqGd_\indexSequence$ \\
\hline
1  & ({+},{-},{i},{j},{+},{-},{-},{-},{+},{+},{+},{+}) & ({-},{+},{-},{+},{+},{-},{+},{+},{j},{j},{+},{+}) \\
2  & ({+},{-},{j},{i},{+},{-},{-},{-},{+},{+},{+},{+}) & ({-},{+},{-},{+},{+},{-},{+},{+},{i},{i},{+},{+}) \\
3  & ({+},{+},{+},{+},{i},{+},{-},{j},{+},{-},{-},{+}) & ({+},{+},{-},{-},{i},{+},{+},{i},{+},{-},{+},{-}) \\
4  & ({+},{-},{-},{+},{i},{-},{+},{j},{+},{+},{+},{+}) & ({-},{+},{-},{+},{j},{+},{+},{j},{-},{-},{+},{+}) \\
5  & ({+},{+},{+},{+},{j},{+},{-},{i},{+},{-},{-},{+}) & ({+},{+},{-},{-},{j},{+},{+},{j},{+},{-},{+},{-}) \\
6  & ({+},{-},{-},{+},{j},{-},{+},{i},{+},{+},{+},{+}) & ({-},{+},{-},{+},{i},{+},{+},{i},{-},{-},{+},{+}) \\
7  & ({+},{+},{+},{+},{i},{-},{+},{j},{+},{-},{-},{+}) & ({+},{+},{-},{-},{i},{-},{-},{i},{+},{-},{+},{-}) \\
8  & ({+},{-},{-},{+},{i},{+},{-},{j},{+},{+},{+},{+}) & ({-},{+},{-},{+},{j},{-},{-},{j},{-},{-},{+},{+}) \\
9  & ({+},{+},{+},{+},{j},{-},{+},{i},{+},{-},{-},{+}) & ({+},{+},{-},{-},{j},{-},{-},{j},{+},{-},{+},{-}) \\
10 & ({+},{-},{-},{+},{j},{+},{-},{i},{+},{+},{+},{+}) & ({-},{+},{-},{+},{i},{-},{-},{i},{-},{-},{+},{+}) \\
11 & ({+},{+},{-},{+},{-},{j},{+},{j},{-},{+},{+},{+}) & ({-},{-},{+},{-},{j},{-},{i},{-},{-},{+},{+},{+}) \\
12 & ({+},{+},{-},{+},{-},{i},{+},{i},{-},{+},{+},{+}) & ({-},{-},{+},{-},{i},{-},{j},{-},{-},{+},{+},{+}) \\
13 & ({+},{+},{+},{-},{-},{i},{-},{j},{-},{+},{-},{-}) & ({+},{+},{+},{-},{j},{+},{j},{-},{+},{-},{+},{+}) \\
14 & ({+},{+},{+},{-},{-},{j},{-},{i},{-},{+},{-},{-}) & ({+},{+},{+},{-},{i},{+},{i},{-},{+},{-},{+},{+}) \\
15 & ({+},{+},{-},{+},{+},{j},{-},{j},{-},{+},{+},{+}) & ({-},{-},{+},{-},{j},{+},{i},{+},{-},{+},{+},{+}) \\
16 & ({+},{+},{-},{+},{+},{i},{-},{i},{-},{+},{+},{+}) & ({-},{-},{+},{-},{i},{+},{j},{+},{-},{+},{+},{+}) \\
17 & ({+},{+},{+},{-},{+},{i},{+},{j},{-},{+},{-},{-}) & ({+},{+},{+},{-},{j},{-},{j},{+},{+},{-},{+},{+}) \\
18 & ({+},{+},{+},{-},{+},{j},{+},{i},{-},{+},{-},{-}) & ({+},{+},{+},{-},{i},{-},{i},{+},{+},{-},{+},{+}) \\
19 & ({+},{+},{+},{-},{i},{i},{+},{-},{+},{+},{-},{+}) & ({+},{+},{+},{-},{-},{-},{j},{i},{-},{-},{+},{-}) \\
20 & ({+},{-},{+},{+},{-},{+},{j},{j},{-},{+},{+},{+}) & ({-},{+},{-},{-},{j},{i},{-},{-},{-},{+},{+},{+}) \\
21 & ({+},{+},{+},{-},{j},{j},{-},{+},{+},{+},{-},{+}) & ({+},{+},{+},{-},{+},{+},{i},{j},{-},{-},{+},{-}) \\
22 & ({+},{-},{+},{+},{+},{-},{i},{i},{-},{+},{+},{+}) & ({-},{+},{-},{-},{i},{j},{+},{+},{-},{+},{+},{+}) \\
23 & ({+},{+},{+},{i},{-},{+},{-},{-},{i},{+},{-},{+}) & ({+},{+},{+},{i},{-},{+},{+},{+},{j},{-},{+},{-}) \\
24 & ({+},{+},{+},{j},{-},{+},{-},{-},{j},{+},{-},{+}) & ({+},{+},{+},{j},{-},{+},{+},{+},{i},{-},{+},{-}) \\
25 & ({+},{+},{-},{+},{+},{+},{j},{i},{-},{-},{-},{+}) & ({+},{+},{-},{+},{j},{j},{+},{-},{+},{+},{+},{-}) \\
26 & ({+},{-},{-},{-},{j},{i},{+},{+},{+},{-},{+},{+}) & ({-},{+},{+},{+},{-},{+},{i},{i},{+},{-},{+},{+}) \\
27 & ({+},{+},{-},{+},{-},{-},{i},{j},{-},{-},{-},{+}) & ({+},{+},{-},{+},{i},{i},{-},{+},{+},{+},{+},{-}) \\
28 & ({+},{-},{-},{-},{i},{j},{-},{-},{+},{-},{+},{+}) & ({-},{+},{+},{+},{+},{-},{j},{j},{+},{-},{+},{+}) \\
29 & ({+},{+},{-},{+},{i},{+},{-},{i},{-},{-},{+},{+}) & ({+},{+},{-},{+},{i},{+},{+},{j},{+},{+},{-},{-}) \\
30 & ({+},{+},{-},{-},{j},{-},{+},{j},{+},{-},{+},{+}) & ({-},{-},{+},{+},{i},{+},{+},{j},{+},{-},{+},{+})
	\end{tabular}
	\vspace{-10pt}
\end{table}

We list the sets obtained for $\seqGc_\indexSequence$ and $\seqGd_\indexSequence$ in Table \ref{table:seqC} when $\thershold=0.715$ and $\upsampleValueU=128$. Based on the aforementioned procedure, 
we could not obtain $\setC$ and $\setD$ when $\thershold<0.715$ for $\numberOfSequences=30$ and $\RBsize=12$. However, the numerical results given in Section~\ref{sec:numerical} show that the maximum cross-correlation is still less than the ones for \ac{ZC} sequences and the sequences adopted in \ac{NR}. It is also worth noting that the sets obtained for $\seqGc_\indexSequence$ and $\seqGd_\indexSequence$ are not unique and depend on the initial seed sequences.

\subsection{Coherent Scheme}
Theorem~\ref{th:golayIterative} allows $\angleGolay[1]$ and $\angleGolay[2]$ to be any unit-norm complex numbers. Hence, 
it indeed enables a scheme which can carry 2 \ac{QPSK} symbols via sequence modulation while inheriting the low \ac{PAPR} benefit of \acp{CS}. By investigating the transmitter given in \figurename~\ref{fig:tx}\subref{fig:nco}, it is straightforward to modulate the sequences on different \acp{RB} by mapping the information bits, i.e., ACK/NACK or \ac{SR}, to $\angleGolay[1]$ and $\angleGolay[2]$, where $\angleGolay[1],\angleGolay[2]\in \setQr\triangleq\{ \constante^{\frac{\constanti\pi}{4}},\constante^{\frac{\constanti3\pi}{4}},\constante^{-\frac{\constanti\pi}{4}},\constante^{-\frac{\constanti3\pi}{4}} \}$. Although this approach has its own merits, the resulting scheme would require another \ac{OFDM} symbol for  channel estimation. On the other hand, some applications with more strict latency constraints may require a framework where a set of \acp{RS} appears in each \ac{RB} for the sake of channel estimation. The question is then if there exists any set of parameters, i.e., $\upsampleVarA, \upsampleVarB, \separationVar$ and initial \acp{GCP}, i.e., $(\seqGa,\seqGb)$ of length $\lengthGaGb$ and $(\seqGc,\seqGd)$  of length $\lengthGcGd$ such that it leads to an interlace with desired resource allocation in frequency while yielding to \acp{RS} in each \ac{RB}.

To address this question, we consider a similar strategy based on Golay's interleaving method and employ Theorem~\ref{th:golayIterative} by choosing $\upsampleVarA = \RBsize\times\numberOfClusters$, $\upsampleVarB=2$, $\separationVar=1$, the initial \acp{GCP}, i.e., $(\seqGa,\seqGb)$ of length $\lengthGaGb=\numberOfClusters$ and $(\seqGc,\seqGd)$ of length $\lengthGcGd=\RBsize/2$, and the elements of $\seqGa$, $\seqGb$, $\seqGc$, and $\seqGd$ be in the set of $\setQ$. Similar to the non-coherent scheme, $\seqGa$ and $\seqGb$ act as spreading sequences for the upsampled sequences $\seqGc$ and $\seqGd$ with the factor of 2, respectively, as
\begin{align}
\polySeq[{\seqGa}][\polyVariable^{\RBsize+\numberOfNulls}]\polySeq[{\seqGc}][\polyVariable^2]=\polySeq[{\upsampleOp[{\RBsize+\numberOfNulls}][\seqGa]*\upsampleOp[2][\seqGc])}][\polyVariable],
\label{eq:partOfCSfir}
\end{align}
and
\begin{align}
\polySeq[{\seqGb}][\polyVariable^{\RBsize+\numberOfNulls}]\polySeq[{\seqGd}][\polyVariable^2]=\polySeq[{\upsampleOp[{\RBsize+\numberOfNulls}][\seqGb]*\upsampleOp[2][\seqGd]}][\polyVariable].
\label{eq:partOfCSsec}
\end{align}
However, unlike the non-coherent scheme, each \ac{RB} is constructed based on the interleaved $\seqGc$ and $\seqGd$ since we choose $\separationVar = 1$ and the size of sequences $\seqGc$ and $\seqGd$ to be half of the \ac{RB} size. Since $\angleGolay[1]$ and $\angleGolay[2]$ are the coefficients of \eqref{eq:partOfCSfir} and \eqref{eq:partOfCSsec} in \eqref{eq:gcpGf}, they are essentially multiplied with the interleaved and spread sequences $\seqGc$ and $\seqGd$ in each \ac{RB}, as illustrated in \figurename~\ref{fig:tx}\subref{fig:co}. Thus, by fixing either $\angleGolay[1]$ and $\angleGolay[2]$, the \acp{RS} on each \ac{RB} can be obtained and the \ac{PAPR} of the responding signal is still always less than or equal to $3$~dB. 
For example, an interlace compatible with the \ac{LTE} \ac{eLAA} can be constructed when $\upsampleVarA =120$, $\upsampleVarB =2$, and $\separationVar = 1$ and the sequences $\seqGa$, $\seqGb$, $\seqGc$, $\seqGd$ are chosen from \cite{holzmann_1991} as $\seqGa = (\constantOne,\constantOne,\constantOne,\constantOne,\constantOne,\constantMinusOne,\constantOne,\constantMinusOne,\constantMinusOne,\constantOne)$, $\seqGb = (\constantOne,\constantOne,\constantMinusOne,\constantMinusOne,\constantOne,\constantOne,\constantOne,\constantMinusOne,\constantOne,\constantMinusOne)$, $\seqGc = (\constantOne,\constantOne,\constantOne,\constanti,\constantMinusOne,\constantOne)$, and $\seqGd = (\constantOne,\constantOne,\constantMinusi,\constantMinusOne,\constantOne,\constantMinusOne)$. Note that this scheme can support multiple access by generating orthogonal codes for  $\seqGc$ and $\seqGd$ through \eqref{eq:shifting} since the resulting sequences via \eqref{eq:shifting}, i.e., $(\seqGc\odot\seqGs,\seqGd\odot\seqGs)$, still construct a \ac{GCP}.

\def\indexRB{k}

\section{Numerical Analysis}\label{sec:numerical}
In this section, we compare the proposed schemes with other schemes in the literature numerically. For the simulations, we consider the \ac{LTE} \ac{eLAA} interlace parameters, i.e., $\RBsize=12$, $\numberOfClusters=10$,  $\numberOfNulls = 9\times12=108$ subcarriers. For the proposed non-coherent scheme, we employ the sequences given in Table~\ref{table:seqC} and the spreading sequences as $\seqGa = (\constantOne,\constantOne,\constantOne,\constantMinusi,\constanti)$, $\seqGb = (\constantOne,\constanti,\constantMinusOne,\constantOne,\constantMinusi)$. For the coherent scheme, we assume that $\seqGa = (\constantOne,\constantOne,\constantOne,\constantOne,\constantOne,\constantMinusOne,\constantOne,\constantMinusOne,\constantMinusOne,\constantOne)$, $\seqGb = (\constantOne,\constantOne,\constantMinusOne,\constantMinusOne,\constantOne,\constantOne,\constantOne,\constantMinusOne,\constantOne,\constantMinusOne)$, $\seqGc = (\constantOne,\constantOne,\constantOne,\constanti,\constantMinusOne,\constantOne)$, and $\seqGd = (\constantOne,\constantOne,\constantMinusi,\constantMinusOne,\constantOne,\constantMinusOne)$. For comparison, we consider \ac{NR} sequences \cite{nr_phy_2017} on each \ac{RB} in the interlace and three \ac{PAPR} minimization methods. The first two methods rely on the optimal phase rotation with \ac{QPSK} alphabet for each \ac{RB} for a given sequence, i.e., \ac{PTS} approach, which prioritize \ac{CM} and \ac{PAPR}. The third approach applies a modulation operation to the sequence on each \ac{RB} as a function of the occupied \ac{RB} index. In other words, the sequence on $\indexRB$th occupied \ac{RB} in the interlace is multiplied with the sequence ${(\exponentialBase^{0\times\indexRB},\exponentialBase^{1\times\indexRB}, \dots, \exponentialBase^{(\RBsize-1)\times\indexRB})}$ for $\indexRB=0,1,\mydots,9$. We refer to this operation as {\em cycling} since the signal component located on each \ac{RB} is cyclically shifted in time. For the fourth design, we generate all possible \ac{ZC} sequences of length 113 (cyclically padded to 120) and select the best $30$ sequences based on the \ac{PAPR} of the corresponding signals after they are mapped to the interlace. For the coherent scheme, we consider \ac{NR} sequences with cycling and modulate  every other subcarrier to transmit ACK/NACK.

\subsection{PAPR/CM Results}
In \figurename~\ref{fig:papr}, the \ac{PAPR} distribution for both coherent and non-coherent schemes are provided. For the alternative non-coherent schemes,  the optimal spreading sequences prioritizing  \ac{PAPR} or \ac{CM} to yield a maximum \ac{PAPR} of $5.3$ dB and $5.7$ dB, respectively, while the \ac{ZC} sequences limits the \ac{PAPR} to $6$ dB. The cycling reduces the maximum \ac{PAPR} to $5.9$ dB and $6.4$ dB for non-coherent and coherent schemes, respectively. On the other hand, both of the proposed non-coherent and coherent schemes offer substantially improved \ac{PAPR} results and limit the \ac{PAPR} to 3 dB as it exploits \acp{CS}. The improvements in terms of \ac{PAPR} are in the range of $2.7$ - $3$ dB and $3.4$~dB as compared to alternative non-coherent and coherent schemes.
\begin{figure}[t]
	\centering
	{\includegraphics[width =3.4in]{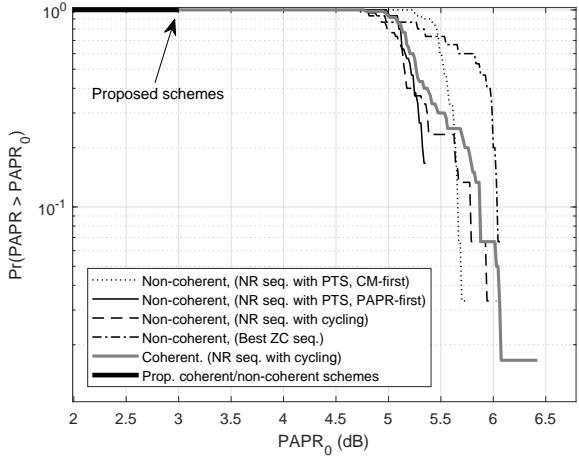}
	\vspace{-10pt}
	}
	\caption{PAPR performance  for different schemes.}
	\label{fig:papr}
\end{figure} 

Another metric that characterizes the fluctuation of the resulting signal is the \ac{CM}. We calculate the \ac{CM} in dB  as
$\cubicMetric = {20 \log_{10}(\operatorRMS[{\vNormCubic}])} / {1.56}$,
where $\vNorm$ is the synthesized signal in time with the power of $1$ \cite{eval_NR}. In \figurename~\ref{fig:cm}, we compare \ac{CM} distributions for the aforementioned schemes. Similar to the \ac{PAPR} results, the proposed schemes improve the \ac{CM} within the range of $0.8$-$1.8$~dB as compared to the alternative approaches, respectively.
\begin{figure}[t]
	\centering
	{\includegraphics[width =3.4in]{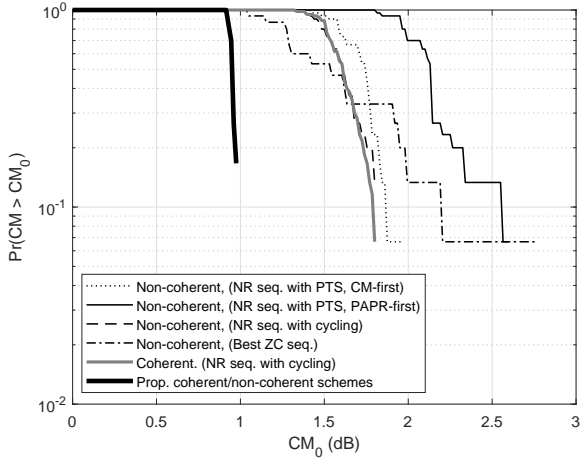}
	}
	\vspace{-10pt}
	\caption{Cubic metric performance for different schemes.}
	\label{fig:cm}
\end{figure}

\def\correlationMetric{\rho}
\def\dftSize{N_{\rm IDFT}}
\def\operationIDFT[#1][#2]{{\rm IDFT}\{#1,#2\}}
\def\operationMAX[#1]{{\rm max}\{#1\}}
\subsection{Peak Cross-correlation Performance}
We evaluate the cross-correlation performance of the sequence designed for the non-coherent scheme by calculating
$    \correlationMetric = {\operationMAX[|{\operationIDFT[\seqGx_\indexSequence\odot\seqGx_j^*][\dftSize]}|]} / {\RBsize},
$
where $\seqGx_\indexSequence$ is the $\indexSequence$th sequence in the set, $\indexSequence\neq j$ and $\operationIDFT[\cdot][\dftSize]$ is the unnormalized \ac{DFT} operation of size $\dftSize$ \cite{eval_NR}. To achieve a large oversampling in time, we choose $\dftSize=4096$. In \figurename~\ref{fig:correlation}, we provide the distribution of $\correlationMetric$ for different schemes. The \ac{ZC} sequences fail as the maximum cross-correlation reaches up to $0.95$ while 50 percentile performance is better than the other methods. The maximum correlation for the set of \ac{NR} sequences rises to $0.8$. On the other hand,
the  correlations for the both sets $\setC$ and $\setD$ are less than or equal to $0.715$ as we set $\thershold=0.715$. Hence, the non-coherent scheme with $\setC$ and $\setD$ provides more robustness against potential inter-cell interference  as compared to the one with \ac{NR} sequences.

\begin{figure}[t]
	\centering
	{\includegraphics[width =3.4in]{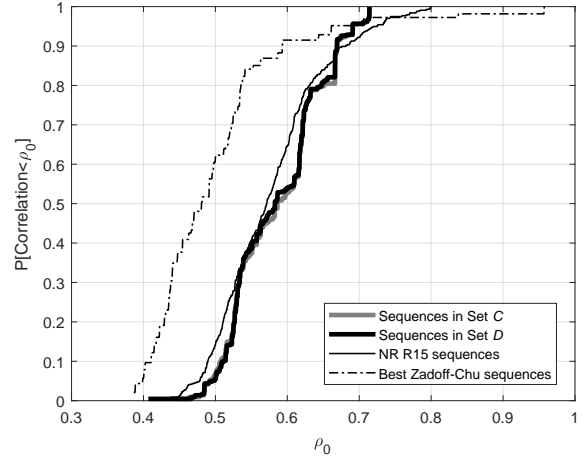}
	}
	\vspace{-10pt}
	\caption{Peak cross-correlation distribution for the sequence sets.}
	\label{fig:correlation}
\end{figure}

\subsection{False Alarm and Miss Detection Results}
In this section, we demonstrate the impact of interlacing on the ACK-to-NACK rate and the ACK miss-detection rate for a given DTX-to-ACK probability, as compared to the single-\ac{RB} approaches. The DTX-to-ACK and NACK-to-ACK rates correspond to the probability of  ACK detection when there is no signal or a NACK is being transmitted, respectively. The ACK miss detection rate is the probability of not detecting ACK when ACK is actually being transmitted. For the single-\ac{RB} approaches, we consider \ac{NR} uplink control channel Format 0, which is a non-coherent scheme, and a coherent scheme which multiplexes \acp{RS} and data in an \ac{RB}. To show the limits, we consider two extreme channel conditions where the occupied \acp{RB} in an interlace experience the same  fading coefficients, i.e., flat fading, or \ac{i.i.d} Rayleigh fading coefficients to model selective fading. In practice, there is always correlation between channel coefficients. However, the correlation can decrease for a large spacing between the occupied \acp{RB} in an interlace. 

In the simulations, we set DTX-to-ACK probability to $1\%$ and consider 2 receive antennas. For the baseband processing, we first detect the energy on the resources. If there is energy, we determine if it is ACK sequence or NACK sequence for non-coherent scheme. For the coherent detection, we estimate the channel and use maximum-ratio combining to combine the symbol energy on each \acp{RB}  to determine if the modulation symbol is ACK or NACK. The results in \figurename~\ref{fig:errorSequencebased} and \figurename~\ref{fig:errorFMDedbased} show that both non-coherent and coherent schemes have the same trends on the NACK-to-ACK and ACK miss-detection rates. In case of flat fading, the interlacing yields results worse than that of the single-RB approaches. This is expected as the baseband processing does not exploit the correlation between the channel coefficients. Otherwise, the performance of the schemes with interlace and single-RB are identical as there is no frequency diversity gain. However, when the channel is frequency-selective, the slopes of the NACK-to-ACK and ACK miss-detection rates change remarkably and the interlacing significantly improves the performance. 

\begin{figure}[t]
	\centering
	{\includegraphics[width =3.4in]{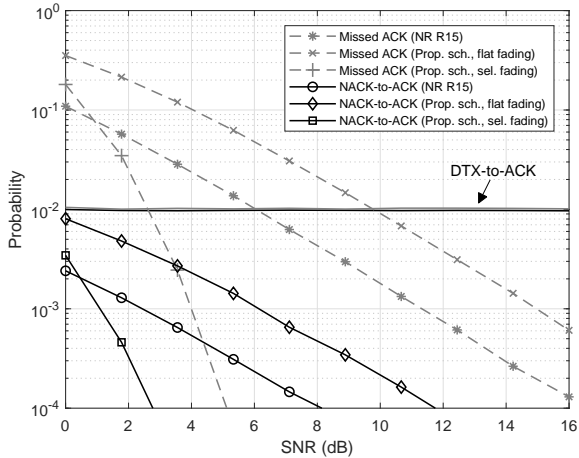}
	}
	\vspace{-10pt}
	\caption{Receiver performance for the proposed non-coherent scheme.}
	\label{fig:errorSequencebased}
\end{figure} 
\begin{figure}[t]
	\centering
	\vspace{-8pt}
	{\includegraphics[width =3.4in]{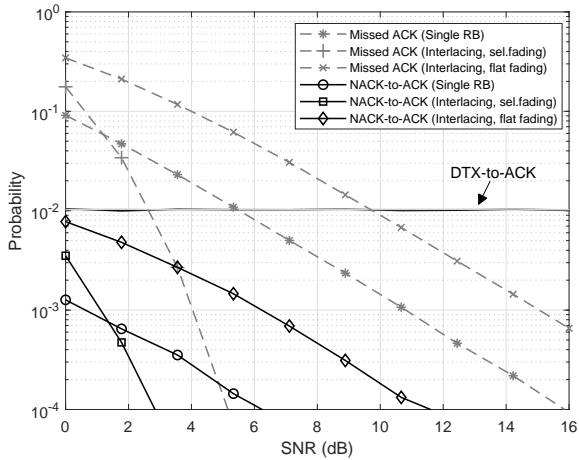}
	}
	\vspace{-10pt}
	\caption{Receiver performance for the proposed coherent scheme.}
	\label{fig:errorFMDedbased}
\end{figure} 

\section{Concluding Remarks}\label{sec:conclusion}
In this study, we  establish a theorem which generalizes Golay's concatenation and interleaving methods to generate non-contiguous \acp{GCP}. We then discuss two schemes for uplink control channels in unlicensed spectrum  by using non-contiguous \acp{GCP}. 
The main benefit of the proposed schemes is that they address the \ac{PAPR} problem of \ac{OFDM} signals while allowing a family of flexible non-contiguous resource allocations. For example, the number of null symbols between the \acp{RB} can be adjusted arbitrarily by using the same pair of spreading sequences and the sequences used in the \acp{RB}. In all cases, the \ac{PAPR} of the corresponding signal is less than or equal to $3$~dB. The \ac{PAPR} gain is around 3~dB as compared to other schemes considered in this study.

While the first scheme can be considered as an extension of the uplink control channel Format 0 in \ac{NR} for the operation in an unlicensed spectrum with non-coherent receivers, the second one achieves the same by enabling coherent detection.
The first scheme separates the \ac{PAPR}  and inter-cell interference minimization problems. While the \ac{PAPR} problem is solved by choosing the sequences for \acp{RB} as a \ac{GCP}, the interference problem is addressed by designing a set of \acp{GCP} for \acp{RB}. With an algorithm which exploits the seed \acp{GCP} provided in \cite{holzmann_1991}, we show that there exists a set of \acp{GCP} of size $30$ and length $12$, which achieves a smaller maximum cross-correlation than that of \ac{NR} sequences.
The second scheme exploits degrees of freedom on the phases provided by Theorem~\ref{th:golayIterative} and leads to  interleaved \acp{RS} and data symbols for each \ac{RB} in an interlace.

In this study, we focus on the schemes which can carry a small amount of information, e.g., ACK/NACK or \ac{SR}, with a single \ac{OFDM} symbol while ensuring low \ac{PAPR} and frequency diversity gain. However, there are many cases where a larger amount of information needs to be transmitted. 
Hence, developing more generic methods with \acp{CS} which can support more bits while achieving low \ac{PAPR} and frequency diversity gain  is not only theoretically interesting, but also practically appealing. We also believe that there may be other applications of non-contiguous \acp{GCP}. Further investigation of non-contiguous \acp{GCP} under different topics is highly encouraged.

\bibliographystyle{IEEEtran}
\bibliography{reliableControl}

\end{document}